\documentclass[a4paper,twoside]{article}

\usepackage{epsfig}
\usepackage{subfigure}
\usepackage{calc}
\usepackage{amssymb}
\usepackage{amstext}
\usepackage{amsmath}
\usepackage{amsthm}
\usepackage{graphicx}
\usepackage{tabularx}
\usepackage{tikz}
\usetikzlibrary{arrows}
\usepackage{verbatim}
\usepackage{enumerate, hyperref}
\usepackage{multicol}
\usepackage{pslatex}
\usepackage{apalike}
\usepackage{SCITEPRESS} 
\usepackage[small]{caption}

\def\leq{\leqslant}

\def\geq{\geqslant}

\newtheorem{theorem}{Theorem}[section]
\newtheorem{definition}{Definition}[section]

\newtheorem{lemma}[theorem]{Lemma}

\newtheorem{corollary}[theorem]{Corollary}
\newtheorem{remark}[theorem]{Remark}

\begin{document}

\title{A Mathematical Model for Signal's Energy at the Output of an Ideal DAC.}
\author{\authorname{Paola Loreti\sup{1} and Pierluigi Vellucci\sup{1}}
\sup{1}Dipartimento di Scienze di Base e Applicate per l'Ingegneria,\\ Via Antonio Scarpa n. 16, 00161 Roma.
\email{paola.loreti@sbai.uniroma1.it., pierluigi.vellucci@sbai.uniroma1.it}
}


\abstract{The presented research work considers a mathematical model for energy of the signal at the output of an ideal DAC, in presence of sampling clock jitter. When sampling clock jitter occurs, the energy of the signal at the output of ideal DAC does not satisfies a Parseval identity. Nevertheless, an estimation of the signal energy is here shown by a direct method involving sinc functions.}

\onecolumn \maketitle \normalsize \vfill

\section{Introduction}
Interpolation based on
\begin{equation}
\label{eq:19}
f(t)=\sum_{n\in\mathbb Z} a_n \operatorname{sinc}(t-n)
\end{equation}
is usually called \emph{ideal bandlimited interpolation}, because it provides a perfect reconstruction for all $t$, if $f(t)$ is bandlimited in $f_m$ and if the sampling frequency $f_s$ is such that $f_s\geq 2f_m$. The \emph{sinc function} in (\ref{eq:19}) is defined as
\begin{equation}
\label{eq:sinc}
 {\operatorname{sinc}}(\alpha)=
   \begin{cases}
   \frac{ \sin(\pi \alpha)}{\pi \alpha}  \qquad &\alpha \not= 0,\\
   1\qquad & \alpha=0.
    \end{cases}
\end{equation}
The system used to implement (\ref{eq:19}), which is known as an \emph{ideal DAC} (i.e. digital-to-analog converter, see \cite{Mano11}), is depicted in block diagram form in figure \ref{fig:1}.

DACs are essential components for measuring instruments (such as arbitrary waveform signal generators) and communication systems (such as transceivers). Since higher sampling speed is being demanded for them, their \emph{sampling clock jitter} effects may be crucial. Jitter is the deviation of a signal's timing event from its intended (ideal) occurrence in time, often in relation to a reference clock source. Therefore, time jitter is an important parameter for determining the performance of digital systems. For a review how time jitter impacts the performance of digital systems, see \cite{Re05}. For digital sampling in analog-to-digital and digital-to-analog converters, it is shown that noise power or multiplicative decorrelation noise generated by sampling clock jitter is a major limitation on the bit resolution (effective number of bits) of these devices, \cite{Re05}.

\begin{figure}[tb]
\centering
\tikzstyle{int}=[draw, fill=blue!20, minimum size=6em]
\tikzstyle{init} = [pin edge={to-,thin,black}]
\begin{tikzpicture}[node distance=2.5cm,auto,>=latex']
    \node [int, pin={[init]above:CLOCK}] (a) {Ideal DAC};
    \node (b) [left of=a,node distance=2cm, coordinate] {a};
    \node  (c) [right of=a] {};
    \node [coordinate] (end) [right of=c, node distance=2cm]{};
    \path[->] (b) edge node {$a_k$} (a);
    \path[->] (a) edge node {$f(t)$} (c);
\end{tikzpicture}
\caption{Representation of the ideal digital-to-analog converter (DAC) or ideal bandlimited interpolator. According to \ref{eq:19}.}
\label{fig:1}
\end{figure}
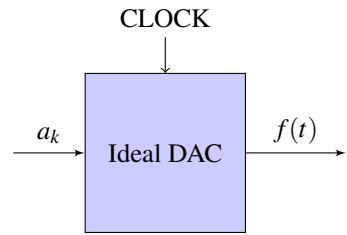

In \cite{Kur02} authors analyze the clock jitter effects on DACs, (Fig. 1 therein), considering a DAC where a digital input is applied with a sampling clock CLK. Ideally the sampling clock CLK operates with a sampling period of $T_s$ for every cycle, however in reality its timing can fluctuate (see Fig. 2 in \cite{Kur02}). Phase and frequency fluctuations have therefore been the subject of numerous studies; well-known references include: \cite{Ab06}, \cite{Demi00}, \cite{Haji98}, \cite{Raza96}.

As it has been well argued in previous works (\cite{Angri09}, \cite{Kur02}), theory dealing with major aspects concerning DAC time base jitter, quantization noise, and nonlinearity is still incomplete; unexpected changes and distortions of waveforms generated via DAC are occasionally supported by simulations and barely investigated by means of experimental activities, \cite{Angri09} and references therein. See also: \cite{Co}, where stochastic analysis is presented in order to predict the average switching rate; \cite{Shi}, where time jittering is modeled as a random variable uniformly distributed; \cite{Ale10}, \cite{Tat}, where jitter effect is assumed as a random variable normally distributed.

In \cite{Angri09}, authors focus on zero-order-hold DACs and, in particular, on how the presence of jitter that can affect their time base modifies the desired features of the analog output waveform. They study more deterministic jitter and develop an analytical model which is capable of describing the spectral content of the analog signal at the output of a DAC, the time base of which suffers from (or is modulated by) sinusoidal jitter. See also: \cite{Gu}, which proposes a first order analytical model describing the influence of the sampling clock modulated by a periodic jitter; \cite{Dap10}, where is introduced a model capable of describing the functioning of a real DAC affected by horizontal quantization, clock modulation, vertical quantization and integral nonlinearity.

In this paper we prove one-sided energy inequality for the output signal of an ideal DAC, in presence of \emph{sampling clock jitter}. Although the energy inequality can be derived for the Fourier transform by the system of complex exponentials \cite{Ing}, here we present a direct proof, based on sinc functions and on the result showed in \cite{Mont}.
Denoting with $f(t)$ the signal, we refer to the following definition of energy.
\begin{definition}
\label{def:1}
The energy in the signal $f(t)$ is
$$E_f:=\int_{-\infty}^\infty |f(t)|^2 dt.$$
\end{definition}
We also denote jitter as $\epsilon_n$, then the $n$-th sampling timing of CLK is $nT_s+\epsilon_n$ instead of $nT_s$. Since we have assumed that $T_s=1$, in our paper sampling timing of CLK is $n+\epsilon_n$ but the results for $T_s\neq 1$ one can obtain in an obvious way. Hence, equation (\ref{eq:19}) becomes
\begin{equation}
\label{eq:a}
f(t)=\sum_{n\in\mathbb Z} a_n \operatorname{sinc}(t-\lambda_n),
\end{equation}
where $\lambda_n=n+\epsilon_n$.
Results obtained here concern a generalization of the Parseval's identity for the sequence of functions $\{\operatorname{sinc}(t-\lambda_n)\}_{n\in\mathbb Z}$, where $\lambda_n\in\mathbb R$. In fact, it is well-known that, for a signal such that
$$f(t)=\sum_{n\in\mathbb Z} a_n \operatorname{sinc}(t-n),$$
its energy is:
$$E_f=\sum_{n\in\mathbb Z} |a_n|^2.$$
This is a Parseval identity for the sequence of functions $\{\operatorname{sinc}(t-n)\}_{n\in\mathbb Z}$, and it is based on the identity
\begin{equation}\label{a2}
\int_{\mathbb R} {\operatorname{sinc}}\big({\tau}-\lambda\big) {\operatorname{sinc}}\big({\tau}-\nu\big)d\tau= {\operatorname{sinc}}(\lambda-\nu ).
\end{equation}
occurred for any real numbers $\lambda$ and $\nu$. But Parseval identity ceases to be true if $n$ is substitutes with $\lambda_n\in\mathbb R$. This motivates the result of the paper, which is described in the following Theorem.
\begin{theorem}
\label{th:main}
Let $I=\{n\, |1\leq n\leq R, \, R\in\mathbb N \}$ be a finite set of integers, and let
\begin{equation}
f(t)=\sum_{n\in I}a_n \operatorname{sinc}(t-\lambda_n),
\end{equation}
where the $\lambda_n$ are real and satisfy
$$|\lambda_n-\lambda_m|\geq \gamma>\sqrt{\frac{1}{3}+\frac{\pi^2}{12}}, \ \ \ \forall n,m\in I.$$
Then
\begin{equation}
E_f\geq \left(1-\gamma^{-1} \sqrt{\frac{1}{3}+\frac{\pi^2}{12}}\right) \sum_{n\in I} |a_n|^2.
\end{equation}
\end{theorem}

\section{Results}
\label{s:1}
For the our purposes, we will use a well-known inequality. \emph{Hilbert's inequality} states that
$$\left|\sum_{n\neq m}\frac{a_n \bar{a}_m}{n-m}\right|\leq\pi \sum_n |a_n|^2$$
for any set of complex $a_n$, where the best possible constant $\pi$ was found by Schur \cite{Sch}. In \cite{Mont} authors obtained a precise bound for the more general bilinear forms:
$$\sum_{n\neq m}a_n\overline a_m\csc\pi(x_r-x_s), \ \ \ \sum_{n\neq m}\dfrac{a_n\overline a_m}{\lambda_r-\lambda_s}.$$
In the following, $\|\theta\|$ denotes the distance from $\theta$ to the nearest integer, that is, $\|\theta\|=\min_n |\theta-n|$. Moreover, $\min_{+} f$ will denotes the least positive value when $f$ ranges over a finite set of non-negative values.
We now give an useful Lemma.
\begin{figure}[tb]
\centering
\includegraphics[scale=0.30]{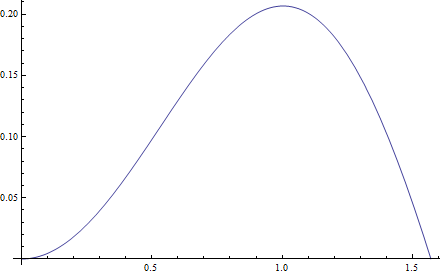}
\caption{Graphic of $g(\theta)=\frac{\pi^2}{4} \sin^2\theta-\theta^2(1+\cos\theta)$ for $\theta\in[0,\pi/2]$.}
\label{fig:fig1}
\end{figure}
\begin{figure}[tb]
\centering
\includegraphics[scale=0.30]{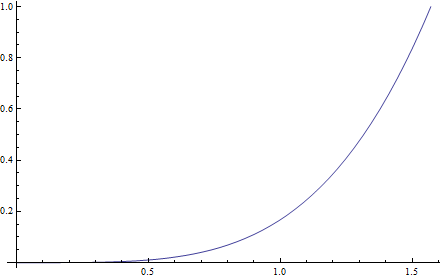}
\caption{Graphic of $\sin^2\theta-\theta^2\cos\theta$ for $\theta\in[0,\pi/2]$.}
\label{fig:fig2}
\end{figure}
\begin{lemma}
\label{l:2}
The inequalities
\begin{equation}
\label{eq:l1}
\csc^2\pi x+|\cot\pi x \csc \pi x|\leq\frac{1}{4}\|x\|^{-2}
\end{equation}
and
\begin{equation}
\label{eq:l2}
|\cot\pi x \csc \pi x|\leq\pi^{-2}\|x\|^{-2}
\end{equation}
hold for all real $x$.
\end{lemma}
\begin{proof}
Let $\theta=\pi x$. We notice that, for an integer $n$, $0\leq\|x\|=\min_n |x-n|\leq \frac{1}{2}$ and so $0\leq \theta \leq \pi/2$. For inequality (\ref{eq:l1}), it is sufficient to show that $g(\theta)\geq0$ in $[0,\pi/2]$, where
$$g(\theta)=\frac{\pi^2}{4} \sin^2\theta-\theta^2(1+\cos\theta).$$
For inequality (\ref{eq:l2}) one shows that:
$$\sin^2\theta-\theta^2\cos\theta\geq0$$
for $\theta\in[0,\pi/2]$. See Figures \ref{fig:fig1} and \ref{fig:fig2}.
\end{proof}
Now we readapt and prove a part of Theorem 1, taken from \cite{Mont}.
\begin{lemma}
\label{th:MV1}
Let $x_1, x_2,... ,x_R$ and $y_1, y_2,... ,y_R$ denote real numbers which are distinct modulo 1, and suppose that
$$\delta=\min_{n,m}{}_{+}\|x_n-y_m\|, \ \ x_n\neq y_m \ \forall n,m=1,\dots,R.$$
Then
\begin{equation}
\label{eq:montvaug}
\left|\sum_{n, m}a_n\overline a_m\csc\pi(x_n-y_m)\right|\leq\delta^{-1}\sqrt{\frac{1}{3}+\frac{\pi^2}{12}}\,\sum_{n=1}^R |a_n|^2.
\end{equation}
where $n$ and $m$ are distinct.
\end{lemma}
\begin{proof}
Our proof is modelled on Montgomery and Vaughan's proof \cite{Mont} of Hilbert's inequality. In \cite{Mont} authors proven that the bilinear form
$$\sum_{n, m}a_n\overline a_m\csc\pi(x_n-x_m),$$
where $n\neq m$, is skew-Hermitian. For this proof we consider the bilinear form:
$$\sum_{n, m}a_n\overline a_m\csc\pi(y_n-y_m)$$
for $n\neq m$. Let us consider
$$\sum_{n}a_n\,\csc\pi(y_n-y_m)= \sum_{n}a_n\,c_{n,m}$$
where $c_{n,m}=\csc\pi(y_n-y_m)$. The RHS is the product of eigenvector $\textbf{a}=(a_1,\dots,a_R)^t$ for the mth column of matrix $C:=(c_{n,m})$. Since the bilinear form under consideration is skew-Hermitian, eigenvalues of matrix $C$ are all purely imaginary or zero, namely there exists a real number $\mu$ such that: $\textbf{a}^t C \textbf{a}=i\mu$. Hence,
\begin{equation}
\label{eq:pv0}
\sum_{n}a_n\,\csc\pi(y_n-y_m)= i\mu a_m
\end{equation}
for $m\neq n$ and $1\leq n,m\leq R$. Also, we may normalize so that $\sum_n |a_n|^2=1$. By Cauchy's inequality,
$$\left|\sum_{n, m}a_n\overline a_m\csc\pi(x_n-y_m)\right|^2\leq \sum_{n}\left| \sum_{m} {}^\prime \bar a_m \csc\pi(x_n-y_m)\right|^2 $$
where $\sum_{m} {}^\prime$ means that all indexes are different. Also,
$$ \sum_{n}\left| \sum_{m} {}^\prime \bar a_m \csc\pi(x_n-y_m)\right|^2 = $$
$$=\sum_{m,p} \bar a_m a_p \sum_{n} {}^\prime \csc\pi(x_n-y_m) \csc\pi(x_n-y_p)$$
\begin{equation}
\label{eq:pv}
=S_1+S_2,
\end{equation}
where
\begin{equation}
\label{eq:pv1}
S_1=\sum_{m} |a_m|^2 \sum_{n} {}^\prime \csc^2\pi(x_n-y_m)
\end{equation}
and
\begin{equation}
\label{eq:pv2}
S_2=\sum_{m\neq p} \bar a_m a_p \sum_{n} {}^\prime \csc\pi(x_n-y_m) \csc\pi(x_n-y_p).
\end{equation}
In $S_2$ we may write
$$\csc\pi(x_n-y_m) \csc\pi(x_n-y_p)=$$
$$=\csc\pi(x_m-y_p)\left[\cot\pi(x_n-y_m)-\cot\pi(x_n-y_p) \right].$$
According to \cite{Mont} (Proof of Theorem 1, p. 79) we use this to split $S_2$ in the following way: $S_2=S_3-S_4+2 \operatorname{Re} S_5$, where
\begin{equation}
\label{eq:pv3}
S_3=\sum_{n,m, p} {}^\prime \bar a_m a_p \csc\pi(y_m-y_p)\,\cot\pi(x_n-y_m),
\end{equation}
\begin{equation}
\label{eq:pv4}
S_4=\sum_{n,m, p} {}^\prime \bar a_m a_p \csc\pi(y_m-y_p)\,\cot\pi(x_n-y_p),
\end{equation}
and
\begin{equation}
\label{eq:pv5}
S_5=\sum_{n,m} {}^\prime \bar a_m a_n \csc\pi(x_n-y_m)\,\cot\pi(x_n-y_m).
\end{equation}
We show now that $S_3=S_4$. We see from (\ref{eq:pv0}) and (\ref{eq:pv3}) that
\begin{align}
S_3& =\sum_{n,m} {}^\prime \bar a_m\, \cot\pi(x_n-y_m) \sum_{p} {}^\prime a_p \csc\pi(y_m-y_p) \notag \\
& =\sum_{n,m} {}^\prime \bar a_m\, \cot\pi(x_n-y_m)  \left(-i\mu a_m\right)\notag \\
& =-i\mu\sum_{n,m} {}^\prime |a_m|^2 \, \cot\pi(x_n-y_m).
\end{align}
Similarly, from (\ref{eq:pv0}) and (\ref{eq:pv4}),
\begin{align}
S_4& =\sum_{n, p} {}^\prime  a_p \cot\pi(x_n-y_p) \sum_{m} {}^\prime \bar a_m\, \csc\pi(y_m-y_p)  \notag \\
& =\sum_{n, p} {}^\prime  a_p \cot\pi(x_n-y_p)  \left(-i\mu \bar a_p\right)\notag \\
& =-i\mu\sum_{n,p} {}^\prime |a_p|^2 \, \cot\pi(x_n-y_p).
\end{align}
Therefore, $S_3=S_4$, so that $S_1+S_2=S_1+2 \operatorname{Re} S_5\leq S_1+2 |S_5|$. We use the inequality $2 |a_n a_m|\leq |a_n|^2+|a_m|^2$ in (\ref{eq:pv5}), so that (\ref{eq:pv1}) and (\ref{eq:pv5}) give
$$S_1+S_2\leq \sum_{m,n} {}^\prime |a_m|^2  \csc^2\pi(x_n-y_m)+$$
$$+\sum_{n,m} {}^\prime \left(|a_n|^2+|a_m|^2\right) \left|\csc\pi(x_n-y_m)\,\cot\pi(x_n-y_m)\right|$$
$$=\sum_{m,n} {}^\prime |a_m|^2 \Bigl(\csc^2\pi(x_n-y_m)+$$
$$+\left|\csc\pi(x_n-y_m)\,\cot\pi(x_n-y_m)\right|\Bigr)+$$
$$+ \sum_{m,n} {}^\prime |a_n|^2\, \left|\csc\pi(x_n-y_m)\,\cot\pi(x_n-y_m)\right|.$$
By Lemma \ref{l:2} this is
$$\leq \frac{1}{4}\sum_m |a_m|^2\sum_n {}^\prime \|x_n-y_m\|^{-2} +$$
$$+ \frac{1}{\pi^2} \sum_n |a_n|^2\sum_m {}^\prime \|x_n-y_m\|^{-2}.$$
A remark similar to that conducted in \cite{Mont}, leads to be conclude that the $x_n$ and the $y_m$ are spaced from each other by at least $\delta$, so that
$$\sum_m {}^\prime \|x_n-y_m\|^{-2}\leq 2 \sum_{k=1}^\infty (k\delta)^{-2}=\frac{\pi^2}{3} \delta^{-2}.$$
Hence,
$$S_1+S_2\leq \frac{\pi^2}{3} \delta^{-2} \left(\frac{1}{\pi^2}+\frac{1}{4}\right)$$
where we have considered $\sum_n |a_n|^2=1$.
\end{proof}

We now able to prove the result of the paper.
\begin{proof}[Proof of Theorem \ref{th:main}]
Put, by hypothesis,
$$\gamma=\min_{n,m}{}_{+}|\lambda_n-\lambda_m|>\sqrt{\frac{1}{3}+\frac{\pi^2}{12}}.$$
Write $\int_{-\infty}^\infty |f(t)|^2 dt$:
$$\sum_{m,n} a_n \bar{a}_m \int_{-\infty}^{+\infty} \operatorname{sinc}(\lambda_n -t)\operatorname{sinc}(\lambda_m -t) dt$$
which is equal to
\begin{equation}
\label{eq:1a}
\sum_n |a_n|^2 + \sum_{m,n} {}^\prime a_n \bar{a}_m \operatorname{sinc}(\lambda_n-\lambda_m).
\end{equation}
Furthermore,
$$\frac{\sin\pi(\lambda_n-\lambda_m)}{\pi(\lambda_n-\lambda_m)}=\frac{1}{\frac{\pi\lambda_n}{\sin\pi(\lambda_n-\lambda_m)}-\frac{\pi\lambda_m}{\sin\pi(\lambda_n-\lambda_m)}}$$
$$=\frac{1}{x_n+x_m}$$
where $x_n:=\frac{\pi\lambda_n}{\sin\pi(\lambda_n-\lambda_m)}$. Putting $y_m=-x_m$ above equality is rewritten as $\frac{1}{x_n-y_m}$, and
$$\sum_{m,n} {}^\prime a_n \bar{a}_m \operatorname{sinc}(\lambda_n-\lambda_m)=\sum_{m,n} {}^\prime \frac{a_n \bar{a}_m}{x_n-y_m} .$$
To prove the Theorem, we note that if $x$ is any member of a bounded interval, then $\|\varepsilon x\| =\varepsilon |x|$ whenever $\varepsilon$ is sufficiently small. Moreover,
$$\frac{1}{x_n-y_m}=\lim_{\varepsilon\rightarrow0}\pi \varepsilon \csc\pi \varepsilon (x_n-y_m)$$
so that we can appeal to Lemma \ref{th:MV1}:
$$\left|\sum_{m,n} {}^\prime \frac{a_n \bar{a}_m}{x_n-y_m} \right|=\pi \varepsilon \left|\sum_{m,n} {}^\prime a_n \bar{a}_m\, \csc\pi \varepsilon (x_n-y_m) \right|$$
$$\leq\pi \varepsilon\, \delta^{-1}\sqrt{\frac{1}{3}+\frac{\pi^2}{12}}\,\sum_{n\in I} |a_n|^2,$$
where, for $\varepsilon\rightarrow0$,
$$\delta=\min_{n,m}{}_{+}\|\varepsilon x_n-\varepsilon y_m\|=\varepsilon\min_{n,m}{}_{+}|x_n-y_m|, $$
$$x_n\neq y_m \ \ \forall n,m=1,\dots,R.$$
Since $x_n:=\frac{\pi\lambda_n}{\sin\pi(\lambda_n-\lambda_m)}$, $y_m=-x_m$,
$$\delta =\varepsilon \min_{n,m}{}_{+}\left|\frac{\pi\lambda_n-\pi\lambda_m}{\sin\pi(\lambda_n-\lambda_m)}\right|$$
and since $\left|\sin\pi(\lambda_n-\lambda_m)\right|\leq 1$, we have
$$
  \delta \geq \varepsilon \pi \min_{n,m}{}_{+} |\lambda_n-\lambda_m|= \varepsilon\pi \gamma
$$
Accordingly,
$$\left|\sum_{m,n} {}^\prime \frac{a_n \bar{a}_m}{x_n-y_m} \right|=\pi \varepsilon \left|\sum_{m,n} {}^\prime a_n \bar{a}_m\, \csc\pi \varepsilon (x_n-y_m) \right|$$
$$\leq\gamma^{-1}\sqrt{\frac{1}{3}+\frac{\pi^2}{12}}\,\sum_{n\in I} |a_n|^2.$$
Thus, an appeal to (\ref{eq:1a}) completes the proof of the Theorem:
$$E_f\geq \left( 1- \gamma^{-1}\sqrt{\frac{1}{3}+\frac{\pi^2}{12}}\right)\sum_{n\in I} |a_n|^2.$$
\end{proof}
As one reads on \cite{Mont}, it follows from a paper of Hellinger and Toeplitz (\cite{Hel} and \cite{Mont}) that Theorem \ref{th:main} and Lemma \ref{th:MV1} hold also for infinite sums, provided that $\min{}_{+} f$ is replaced by $\inf{}_{+} f$. It is also possible to consider bilateral series if we put $\lambda_{-n}=-\lambda_n$ for $n=1,2,\dots$.

An estimate from above is immediate employing same steps involved used in the proof of theorem \ref{th:main}. Indeed, from equation (\ref{eq:1a}) and by triangle inequality:
$$\sum_n |a_n|^2 + \sum_{m,n} {}^\prime a_n \bar{a}_m \operatorname{sinc}(\lambda_n-\lambda_m)\leq $$
$$\leq\left(1+\gamma^{-1}\sqrt{\frac{1}{3}+\frac{\pi^2}{12}}\right) \sum_n |a_n|^2$$
where $\gamma$ is defined as in theorem \ref{th:main}.
\begin{corollary}
\label{cor:main}
Let $I=\{n\, |1\leq n\leq R, \, R\in\mathbb N \}$ be a finite set of integers, and let
\begin{equation}
f(t)=\sum_{n\in I}a_n \operatorname{sinc}(t-\lambda_n),
\end{equation}
where the $\lambda_n$ are real and satisfy
$$|\lambda_n-\lambda_m|\geq \gamma>\sqrt{\frac{1}{3}+\frac{\pi^2}{12}}, \ \ \ \forall n,m\in I.$$
Then
\begin{equation}
E_f\asymp \sum_{n\in I} |a_n|^2.
\end{equation}
\end{corollary}
\begin{remark}
Write $E_f \asymp \sum_{n\in I} |a_n|^2$ means that
$$c_1 \sum_{n\in I} |a_n|^2 \leq E_f\leq c_2 \sum_{n\in I} |a_n|^2$$
with two constants $c_1, c_2 > 0$, independent of the particular form of $f(t)$, except for the assumption $|\lambda_n-\lambda_m|\geq \gamma>\sqrt{\frac{1}{3}+\frac{\pi^2}{12}}$, $\forall n,m\in I$.
\end{remark}




\vfill
\bibliographystyle{apalike}
{\small
\bibliography{example}}

\vfill
\end{document}